\documentclass[11pt]{article}
\usepackage[utf8]{inputenc}
\usepackage[T1]{fontenc}
\usepackage[top=1in, bottom=1in, left=1in, right=1in, letterpaper]{geometry}

\usepackage{amsthm, amsmath, amssymb, amsfonts, bm}

\usepackage{enumitem}
\usepackage{comment}
\usepackage{xspace}
\usepackage{ifthen}
\usepackage{boxedminipage}
\usepackage{algorithm}
\usepackage{algpseudocode}
\usepackage{color}
\usepackage{empheq}
\usepackage[colorlinks=true,citecolor=blue,bookmarksnumbered=true,hyperfootnotes=true]{hyperref}
\usepackage{cleveref}

\algrenewcomment[1]{\hfill{\scriptsize//~#1}}

\newtheorem{theorem}{Theorem}
\newtheorem{lemma}{Lemma}

\title{Maximising the Influence of Temporary Participants in Opinion Formation}

\author{Zhiqiang Zhuang\thanks{College of Intelligence and Computing, Tianjin University, China} \and Kewen Wang\thanks{School of Information and Communication Technology, Griffith University, Australia} \and Zhe Wang\footnotemark[2] \and Junhu Wang\footnotemark[2] \and Yinong Yang\thanks{Department of Mathematics, Liaoning University, China}}

\date{}

\begin{document}

\maketitle

\begin{abstract}
DeGroot-style opinion formation 
presume a continuous interaction among agents of a social network.
Hence, it cannot handle agents external to the social network 
that interact only temporarily with the permanent ones.
Many real-world organisations and individuals fall into such a category.
For instance, a company tries to persuade as many as possible to buy its products
and, due to various constraints, can only exert its influence for a limited
amount of time.
We propose a variant of the DeGroot model that allows an external 
agent to interact with the permanent ones for a preset period of time.
We obtain several insights on maximising an external agent's influence
in opinion formation by analysing and simulating the variant.

\end{abstract}

\section{Introduction}

Consider a social network in which people have an opinion about the state of something 
in the world, such as the willingness to buy a product, 
the effectiveness of a public policy, or the reliability of an economic forecast.
Rather than forming opinions on their own, 
people tend to learn about the state of the world via observation and communication with others. Mathematical models of \textit{opinion formation} try to formalize these interactions 
by describing how people process the other's opinions 
and how their opinions evolve as a result of the interactions \cite{book/Jackson08}.

The \textit{DeGroot model} \cite{DeGroot} is a benchmark opinion formation model that 
has found usage in many disciplines.
The model describes a discrete time opinion formation process in which
agents within a social network have an initial opinion 
that they update by repeatedly taking weighted average of their friends' opinions.
Over the years, some highly influential variants of the DeGroot model have been proposed 
to take into account real-world situations that were neglected in the original model
such as the \textit{Friedkin and Johnson model} \cite{journal/JMS/Friedkin1990} and
the \textit{bounded confidence model} \cite{journal/Deffuant2000,journal/JASSS/Hegselmann2002}.
A recurrent topic in DeGroot-style opinion formation
is the identification of conditions for reaching a consensus 
and the quantification of individual influence in forming the consensus  \cite{journal/ARC/PROSKURNIKOV2017}.

Implicit in the DeGroot model and its variants is the assumption that 
the agents of a social network interact continuously 
where no agent skips any interaction at any time.
This is a reasonable assumption, given the dynamic nature of opinion formation 
and the research focus on its limiting behaviour.
However, it excludes external agents that do not have a permanent presence 
in the social network but may have a considerable influence to the permanent agents.
A prominent example is an organisation
trying to persuade people to for instance buy its products or
vote for a particular candidate through advertising.
Due to constraints like budget and timing, the organisation 
can advertise or exert their influence only for a limited amount of time,
nevertheless, for some people, the organisation's influence is at least 
comparable to that of their friends in the social network.

Traditionally, agents who refused to be influenced by others
are termed as \textit{stubborn} agents (also called \textit{zealots} \cite{journal/NJP/Masuda_2015} or \textit{radicals} \cite{journal/NHM/Hegselmann2015}) \cite{journal/ARC/PROSKURNIKOV2017}. 
Assuming a permanent presence,
eventually they make every non-stubborn agents submissive of their opinions. 
The aforementioned external agents are also stubborn for having uncompromising opinions, 
however, a sharp difference with the orthodox modeling is their temporary nature.   
To the best of our knowledge we are the first to consider stubborn agents 
without a permanent presence. This setting not only is more realistic, but also 
opens up new perspectives to investigate behaviors of such agents.
For instance, with a limited time frame to exert influence, 
these agents are confronted with the strategic problem
of how to allocate their resource to achieve maximum influence.

Our goal in this paper is twofold.
Firstly, we propose a variant of the DeGroot model that 
allows an external (stubborn) agent to participate only temporarily.
The variant enables the investigation on
how combinations of the following four factors affect the external agent's influence.
We illustrate them in the context of 
a company promoting its products by TV commercials.

\begin{description}
\item[Coverage:] the number of agents to which
the external agent can exert its influence. 
This is the expected number of viewers of the TV commercial each time it is broadcast. 
\item[Duration:] the number of times
the external agent can exert its influence.
This is the expected number of times the TV commercial is broadcast.
\item[Intensity:] the amount of influence 
the external agent can exert to other agents each time it does so.
This reflects the TV commercial's impact to its viewers each time they see it.
\item[Timing:] the time points at which the external agent
exerts its influence.
This reflects the time points at which the TV commercial is broadcast.
\end{description}

\noindent
Secondly, we articulate insights on how to allocate the external agents' resources to the four factors in order to maximise its influence 
through mathematical analysis and computer simulation.
According to our analysis and simulations, 
the timing factor is irrelevant if the coverage factor is at its maximum (i.e., full coverage); the coverage and the duration factor are equally important; and 
it is more effective to allocate resource
to scale up the intensity factor than to scale up the duration factor.
We also derive several other insights that deepen our understanding of opinion formation in general.

After giving some preliminaries, 
we present our model that incorporates the aforementioned factors 
into the opinion formation process.
This is followed by insights obtained by analytical method and subsequently simulations.  

\section{Preliminaries}

In this paper, we write matrices as uppercase letters in boldface
such as $\mathbf{T}$,
vectors as lowercase letters in boldface
such as $\mathbf{p}$, and
scalars as lowercase letters such as $m$.
We denote the set of real numbers and integers as $\mathbb{R}$ and $\mathbb{Z}$ respectively.
The entry in the $i$-th row and $j$-th column
of a matrix $\mathbf{T}$ is denoted as $t_{ij}$.
The \emph{transpose} of a matrix $\mathbf{T}$ is written as $\mathbf{T}^\top$. 
A matrix is \emph{non-negative} if all its entries are non-negative.
A non-negative matrix $\mathbf{T}$ is \emph{stochastic} 
if all its rows sum to 1, that is
$\sum_j t_{ij}=1$ for all $i$.
Vectors are considered as single column matrices
unless otherwise specified. 
The $i$th component of a vector $\mathbf{p}$ is denoted 
as $p_i$. 
The ``zero'' vector $(0,\ldots,0)^\top$ and the ``one'' vector $(1,\ldots,1)^\top$ are denoted as $\mathbf{0}$ and $\mathbf{1}$ respectively, and are with dimensions suitable to the context they appear.

A directed graph is a pair $(V,E)$ where 
$V$ is the set of \textit{nodes} and $E\subseteq V\times V$ the set of \emph{edges}.
In opinion formation models, 
a directed graph $(V,E)$ is often identified by its adjacency matrix
$\mathbf{T}$ which is a non-negative matrix such that $(i,j)\in E$ iff $t_{ij}>0$. 
\textit{Paths}, \textit{cycles} and and their \textit{lengths} 
in a directed graph are defined in the standard way.
A directed graph or equivalently an adjacency matrix 
is \emph{strongly connected} if there is a path 
from any node to any other node and it is \emph{aperiodic} 
if the greatest common divisor of the lengths of its cycles is one.

\section{Models of Opinion Formation}

In this section, we present our opinion formation model.
It involves $n$ permanent agents interacting continuously and
an external agent interacting temporarily with the permanent ones.
Outside the external agent's period of interaction,
our model reduces to the DeGroot Model in which 
the pattern of interaction is represented as
a $n\times n$ stochastic matrix $\mathbf{T}$.
An entry $t_{ij}$ of $\mathbf{T}$ represents the weight agent $i$ places
on agent $j$. The weights $t_{ij}$ for $j=1,2,\ldots,n$ are considered as finite resources,
distributed by $i$ to itself and others.
A positive $t_{ij}$ indicates $j$ is able to influence $i$
at each round of interaction; and the greater $t_{ij}$ is, 
the stronger the influence.
We refer to $\mathbf{T}$ as the \emph{interaction matrix} 
which can be seen as the adjacency matrix that captures
the social network structure.
Conventionally, opinions are represented as real numbers 
and time is measured in rounds of interactions.
We denote agent $i$'s opinion and the vector of the $n$ agents' opinions
after the $t$th round of interaction as $p_i^{(t)}$ and $\mathbf{p}^{(t)}$ respectively.
Moreover, we refer to $p_i^{(0)}$ as the \emph{initial opinion} of $i$ and $\mathbf{p}^{(0)}$ the \emph{initial opinion vector}.
The agents repeatedly interact by synchronously taking weighted averages
of the opinions of agents who can influence them,
that is $\mathbf{p}^{(t)}$ obeys the following equation for $t\geq 1$
\begin{equation}\label{eq:DeGroot}
\mathbf{p}^{(t)}=\mathbf{T}^t \mathbf{p}^{(0)}.
\end{equation}
The weight $t_{ij}$ is therefore the contribution of $j$'s opinion at each round of interaction to $i$'s opinion at the next round.
The weight $t_{ii}$ agent $i$ places on itself represents its openness to other agent's influence:
$t_{ii}=0$ indicates an open-minded agent who totally relies
on the others' opinions whereas $t_{ii}=1$ indicates a \emph{stubborn} agent  
whose opinion remains unchanged.

An opinion formation (as described by Equation~(\ref{eq:DeGroot})) is \emph{convergent} if 
$$\mathbf{p}^\infty=\lim_{t\rightarrow \infty}\mathbf{T}^t \mathbf{p}^{(0)}$$
exists for any $\mathbf{p}^{(0)}$.
A convergent opinion formation reaches a \emph{consensus} if 
all components of $\mathbf{p}^\infty$ are identical,
which happens when all rows of $\lim_{t\rightarrow \infty}\mathbf{T}^t$ are identical.
We refer to  $\mathbf{p}^\infty$ as the \emph{limiting opinion vector} and $p_i^\infty$ the limiting opinion of $i$.
It is shown that if the interaction matrix is strongly connected,
then a convergent opinion formation always reaches a consensus.
Moreover, an opinion formation with a strongly connected interaction matrix $\mathbf{T}$
is convergent iff $\mathbf{T}$ is aperiodic
or equivalently there is a unique left eigenvector $\mathbf{s}$ of $\mathbf{T}$, 
corresponding to eigenvalue 1 such that $\sum_{i=1}^n s_i=1$ and 
\begin{equation}\label{eq:influence_vector}
p_i^\infty=\left( \lim_{t\rightarrow\infty}\mathbf{T}^t\mathbf{p}\right)_i=\mathbf{sp}^{(0)}
\end{equation}
for $i=1,2,\ldots,n$.
The result establishes whether an opinion formation converges 
and what it converges to when it does.
Also the result 
implies $\lim_{t\rightarrow\infty}\mathbf{T}^t$ is a matrix with identical rows each of which is the unique left eigenvector $\mathbf{s}$.
As $p_i^\infty$ are identical for $i=1,2,\ldots,n$,
we refer to all of them as the \emph{limiting opinion}.  
See the survey \cite{journal/ARC/PROSKURNIKOV2017} for the other convergence and consensus conditions and \cite{book/Meyer} for the technical details on matrix.

According to Equation~(\ref{eq:influence_vector}), the limiting opinion
is a weighted average of the initial opinions where agent $i$'s weight is $s_i$. 
These weights are commonly taken as the measure of an agent's influence 
in a DeGroot-style opinion formation and are sometimes referred to as the agents' \textit{social influence} \cite{DeGroot,book/Jackson08}.
We refer to $\mathbf{s}$ as the \emph{social influence vector}.
Note that, $\mathbf{s}$ being a left eigenvector of $\mathbf{T}$ with an eigenvalue of 1 means $\mathbf{s}\mathbf{T}=\mathbf{s}$ which implies $s_i=\sum_{j=1}^n t_{ji}s_j$. 
Thus the social influence of $i$ is a weighted sum of the social influences of the various agents who can be influenced by $i$.
This is a very natural property of a measure of influence and entails that an influential person is one who is trusted by other influential persons. 
We adopt this measure of influence to quantify the external agent's 
influence in opinion formation.

The novelty of our model lies in the treatment of the external agent
that interacts with the permanent ones for a finite number of rounds.
We reserve the letter $k$ for this finite number which indicates 
the \textit{duration} of the external agent's influence.
To represent the pattern of interaction involving the external agent,
we extend the interaction matrix $\mathbf{T}$
to form a $(n+1)\times (n+1)$ interaction matrix $\mathbf{A}$
in which the external agent is identified as the $(n+1)$th agent 
whose corresponding weights occupy the $(n+1)$th row and column.
We refer to $\mathbf{A}$ as the \emph{extended interaction matrix}.
Since the external agent acts as an organisation or individual 
with the solitary goal of persuading the others of its opinion,
the external agent is, in our terminology, a stubborn agent, hence $a_{(n+1)(n+1)}=1$.
For the rest of the entries in the $(n+1)$th column,
$a_{i(n+1)}>0$ means the external agent is able to influence agent $i$.
We reserve the letter $m$ for the number of such positive entries  
which indicates the \textit{coverage} of the external agent's influence.
A simplifying assumption we make is 
that the $m>0$ agents place an identical weight of $\lambda$ to the external agent, 
that is $a_{i(n+1)}=\lambda$ whenever $a_{i(n+1)}>0$ and $i\leq n$.
The weight $\lambda$ indicates the \textit{intensity} of the external agent's influence.
We also assume, without loss of generality, that the $m$ agents
occupy the first $m$ rows and columns of $\mathbf{A}$.
We reserve the letter $\Lambda$ for the vector 
that occupies the first $n$ entries of the $(n+1)$th column.
Lastly, for each entry $a_{ij}$ with $1\leq i,j\leq n$, 
if $a_{i(n+1)}=0$,  then it inherit the corresponding entry $t_{ij}$ in $\mathbf{T}$,
otherwise it is shrank from $t_{ij}$ by a factor of $(1-\lambda)$ 
to make $\mathbf{A}$ stochastic.
Putting these together, we have
$${\mathbf{A} ={\begin{bmatrix}(1-\lambda)  t_{11}&\cdots &(1-\lambda)t_{1n}&\lambda\\\vdots &\ddots &\vdots&\vdots \\(1-\lambda) t_{m1}&\cdots &(1-\lambda) t_{mn}&\lambda\\
t_{(m+1)1}&\cdots & t_{(m+1) n}&0\\
\vdots &\ddots &\vdots&\vdots \\
t_{n1}&\cdots & t_{nn}&0\\
0&\cdots &0&1\end{bmatrix}}}$$
where $m$ is the number of agents the external agent can influence,
$t_{ij}$ are entries of $\mathbf{T}$, 
and $\lambda$ is the weight placed on the external agent.

The $n\hspace{-0.4mm}+\hspace{-0.4mm}1$ agents interact exactly as in the DeGroot model
only that it is now governed by both $\mathbf{T}$ and $\mathbf{A}$.
The opinion vector $\mathbf{p}^{(t)}$ obeys Equation~(\ref{eq:DeGroot})
when the external agent does not participate
and the following one when it does.
\begin{equation}\label{eq:DeGroot_tp}
\mathbf{p}^{(t)}=
     \left(\mathbf{A} \left(\begin{array}{c}
     \mathbf{p}^{(t-1)}\\
     a
    \end{array}\right)\right)_{1,\ldots,n}
\end{equation}
\noindent
where $a$ is the external agent's unchanged opinion.
For a matrix $\mathbf{W}$,
we denote the matrix formed by the first $n$ rows
of $\mathbf{W}$ as $(\mathbf{W})_{1,\ldots,n}$.
For example, $\big(\begin{smallmatrix}
     a & b \\
     c & d \\
     e & f
\end{smallmatrix}\big)_{1,2}=
    \big(\begin{smallmatrix}    
     a & b \\
     c & d 
\end{smallmatrix}\big)$.
To illustrate this interaction, suppose the external agent 
participates in the third and forth rounds of interaction,
then 
$$\mathbf{p}^{(4)}= \left(\mathbf{A}^2 \left(\begin{array}{c}
     \mathbf{T}^2\mathbf{p}^{(0)}\\
     a
    \end{array}\right)\right)_{1,\ldots,n}$$
and $\mathbf{p}^\infty=\lim_{t\rightarrow \infty}\mathbf{T}^t \mathbf{p}^{(4)}$.

Following Equation~(\ref{eq:influence_vector}), after expressing the limiting opinion 
as the weighted average of the $n\hspace{-0.4mm}+\hspace{-0.4mm}1$ agents' initial opinions, that is
\begin{equation}\label{eq:soical_influence}
p_i^\infty=w_1 p_1^{(0)}+w_2 p_2^{(0)}+\cdots+w_n p_n^{(0)}+w_{(n+1)} a
\end{equation}
we take weight $w_{(n+1)}$ as the external agent's social influence.
Due to the external agent's intervention,
the social influence vector $\mathbf{s}$ no longer 
gives the accurate social influence for the permanent agents.
It does so if the external agent did not participate at all in which case our model reduces to the DeGroot model.

For the rest of this paper, we assume all interaction matrices are strongly connected
and aperiodic to ensure an opinion formation always reaches a consensus, for otherwise
our measure of influence cannot be defined.
Also we assume, without loss of generality, 
that the external agent's unchanged opinion is $1$ 
and every permanent agent has an initial opinions of $0$.
Note that agents' opinions are irrelevant to their influence
and we only concern with the external agent's influence.
By this assumption, it follows from Equation~(\ref{eq:soical_influence}) that
the limiting opinion is precisely the external agent's social influence.

\section{Analytical Results}

We have argued that an organisation's influencing effort
depends on the coverage, duration, intensity, and timing of its influence.
Our model captures these factors respectively 
as the number $m$ of agents the external agent can influence;
the number $k$ of rounds of interactions the external agent participates;
the weight $\lambda$ that is placed on the external agent; and 
the time points in which the $k$ rounds of interaction take place. 
Of those factors, the first three reflect the amount of resource 
available for the influencing effort, 
the more resource there is, the larger these factors' values.
As organisations usually, if not always, have a finite resource,  
a vital question is how and when to allocate it for achieving maximum influence.  

Devoted to analytical results, in this section, 
we express the external agent's influence as
a function of the dependent factors and obtain influence maximising insights 
by analysing the function.
For the ease of presentation, 
we decompose the multiplication with the matrix $\mathbf{A}$ 
in Equation~(\ref{eq:DeGroot_tp}) and rewrite it as
\begin{equation}\label{eq:update_partial}
  \mathbf{p}^{(t)}=(\mathbf{T}-\lambda (\mathbf{T})_m)\mathbf{p}^{(t-1)}+\Lambda.
\end{equation}
\noindent
Recall that we assume $a = 1$ and
$\Lambda$ is the n-dimensional vector of which
the first $m$ components are $\lambda$ and the rest are $0$.     
For a matrix $\mathbf{T}$, 
$(\mathbf{T})_m$ is the matrix formed by replacing all entries of $\mathbf{T}$ with zero except those of its first $m$ rows.
For example  $\big(\begin{smallmatrix}
     a & b \\
     c & d 
\end{smallmatrix}\big)_{1}=
    \big(\begin{smallmatrix}    
     a & b \\
     0 & 0 
\end{smallmatrix}\big)$.

One might have noticed that the timing factor 
is intrinsically different from the other three.
Apart from being unaffected by the scarcity of resource, more importantly,
the timing factor is not a single factor, 
but a collection of $k$ factors,
one for each of the $k$ rounds of participation.
It is impractical and pointless to consider all variations of the $k$ factors,
instead we focus on three timing options
that echo real-world situations:
(1) the external agent only participates after 
the other agents have reached a consensus; 
(2) the external agent participates for the first $k$ rounds of interaction;
and (3) the external agent randomly participates for $k$ rounds of interaction
according to a uniform distribution over a range of time points.
And we refer to them as 
\textit{consensus}, \textit{start}, and \textit{uniform} respectively. 
For now we only concern with the consensus timing option.
Later on we will explain how the other two 
cause the explosion in the number of variables 
that makes function analysis infeasible.
In fact, ``consensus'' is the most important of the three as 
simulations show that it  gives rise to
the largest social influence for the external agent.

The consensus timing option resembles 
the real-world situation in which an organisation always allow sufficient time for 
its target to thoroughly digest its influence from the previous intervention
before intervening and exerting its influence again. 
From a function analysis perspective, 
it keeps the number of variables small which results in 
a succinct expression for the limiting opinion.

\begin{theorem}\label{thm:limiting_partial}
In an opinion formation, if the external agent can influence $m\leq n$ agents 
and it participates in $k\geq 1$ rounds of interaction, each of which is
at a time point the other agents have reached a consensus, then 
$$p_i^\infty=\textstyle\sum_{j=0}^{k-1}(1-s\lambda)^j s \lambda$$
for all $i$,
where $s=\sum_{i=1}^m s_i$ for $\mathbf{s}$ the social influence vector.
\end{theorem}
\begin{proof}
Let $\mathbf{S}=\lim_{t\rightarrow\infty} \mathbf{T}^t$, 
so each row of $\mathbf{S}$ is
the social influence vector $\mathbf{s}$.
We prove by induction on $k$ that 
$\mathbf{p}^\infty=\sum_{j=0}^{k-1}(1-s\lambda)^j (s\lambda, \ldots, s\lambda)^\top$
where $s=\sum_{i=1}^m s_i$.
For the base case of $k=1$.
It follows from Equation~(\ref{eq:DeGroot}) and (\ref{eq:update_partial}) that
\begin{align*}
 \mathbf{p}^\infty &= \mathbf{S} ((\mathbf{T}-\lambda (\mathbf{T})_m)\mathbf{0}+\Lambda)\\
 &= \mathbf{S} (\lambda,\ldots,\lambda,0,\ldots,0)^\top\\
 &= (\textstyle\sum_{i=1}^m s_i \lambda, \ldots, \textstyle\sum_{i=1}^m s_i \lambda)^\top\\
 &=(s\lambda, \ldots, s\lambda)^\top
\end{align*}
For the induction step, suppose 
$\mathbf{p}^\infty=\sum_{j=0}^{k-1}(1-s\lambda)^j (s\lambda, \ldots, s\lambda)^\top$
for $k=l$.
We need to show the equality also holds for $k=l+1$.
Let $\mathbf{a}$ be the limiting opinion vector
for when the external agent participates for $l$ rounds of interaction.
Due to the induction hypothesis 
$\mathbf{a}=\sum_{j=0}^{l-1}(1-s\lambda)^j (s\lambda, \ldots, s\lambda)^\top$.
Then for $k=l+1$ we have
\begin{align*}
\mathbf{p}^\infty
&=\mathbf{S}((\mathbf{T}-\lambda (\mathbf{T})_m)\mathbf{a}+\Lambda) \\
&=\mathbf{S}\mathbf{T}\mathbf{a}-\lambda\mathbf{S}(\mathbf{T})_m\mathbf{a} +\mathbf{S}\Lambda\\
&=\mathbf{a}-\lambda \mathbf{S}(\mathbf{T}\mathbf{a})_m+\mathbf{S}\Lambda \\
&=\mathbf{a}-\lambda \mathbf{S}(\mathbf{a})_m+(s\lambda, \ldots, s\lambda)^\top\\
&=\mathbf{a}-\lambda (\textstyle\sum_{i=1}^m s_i)\mathbf{a} +(s\lambda, \ldots, s\lambda)^\top\\
&=(1-s\lambda)\mathbf{a}+(s\lambda, \ldots, s\lambda)^\top \\
&=\textstyle\sum_{j=0}^{l}(1-s\lambda)^j (s\lambda, \ldots, s\lambda)^\top
\end{align*}
This complete the induction step.
Thus we have $p_i^\infty=\sum_{j=0}^{k-1}(1-s\lambda)^j s \lambda$ for all $i$.

\end{proof}

\noindent
Essentially, the limiting opinion is the sum of the first $k$ terms of a geometric series with start term $s\lambda$ and constant ratio $1-s\lambda$ 
where the scalar $s$ is the sum of the first $m$ components of the social influence vector $\mathbf{s}$. 
Since our model does not specify the identities of the $m$ agents that can be influenced by the external agent, 
the same value of $m$ can lead to different values of $s$.
So ultimately it is the value of $s$ rather than $m$ that matters
in determining the limiting opinion.

We can represent the limiting opinion, 
which is the external agent's social influence, as a function
$I:X\times Y\times Y\rightarrow\mathbb{R}$ such that 
$$I(k,\lambda,s)=\textstyle\sum_{j=0}^{k-1}(1-s\lambda)^j s \lambda$$
where $X=\{x\in \mathbb{Z}\,|\,x\geq 1\}$ and $Y=\{y\in \mathbb{R}\,|\,0<y<1\}$.
As explained, the variable $m$ should not appear in the function and
$s$ is the combined social influence of the $m$ agents 
for when the external agent did not participate.
Substituting $I(k,\lambda,s)$ with the formula for the sum of a geometric series, we have
\begin{equation}\label{eq:geometric}
    I(k,\lambda,s)=1-(1-s\lambda)^k
\end{equation}
It is easy to see that $I(k,\lambda,s)$ increases as $s$ gets lager.
Since a few influential agents 
can have the same combined social influence as 
that of many less influential ones,
a large value of $m$ does not necessarily give a large value of $s$.
This leads to our first insight through function analysis:
\begin{quote}
\textit{(A1) if the timing option is consensus,
then the external agent should aim for the more influential ones to maximise its social influence.}
\end{quote} 
What this analysis also tells us is that the number of agents that can be
influenced by the external agent is not the most accurate measure of ``coverage.''
A better choice is the combined social influence of such agents.

With the function $I(k,\lambda,s)$,
we note that its three variables
have nothing to do with the
structure of the social network, which means the latter has no
impact on the external agent’s social influence. 
This leads to our second insight:

\begin{quote}
\textit{(A2) if the timing option is consensus,
then the structure of the social network is irrelevant to the external agent's social
influence.}
\end{quote} 

\noindent
The insight might seem trivial, 
nevertheless it is of great significance in practice.
Often a social network's structure 
is unknown to an external organisation, 
so knowing that it is irrelevant brings certainty and assurance
to an organisation's influencing effort.
Hence, we consider this irrelevance as an advantage of the ``consensus'' timing option over the others with which 
the structure does matter. 

With $I(k,\lambda,s)$, we also note that
the extra influence accumulated by participating one more round
of interaction is $I(k+1,\lambda,s)-I(k,\lambda,s)$.
Substituting $I(k+1,\lambda,s)$ and $I(k,\lambda,s)$ with their full expression, 
we have $$I(k+1,\lambda,s)-I(k,\lambda,s) =(1-s\lambda)^k s \lambda.$$
Since $0<1-s\lambda<1$, $I(k+1,\lambda,s)-I(k,\lambda,s)$ decreases as $k$ gets larger.
This leads to our third insight:
\begin{quote}
\textit{(A3) if the timing option is consensus,
then the external agent's influencing effort become less and less effective
as it participates in more rounds of interaction.}
\end{quote}

Next, we will analyse $I(k,\lambda,s)$
to decide which of the three variables $s$, $k$, and $\lambda$
has a more profound effect on the external agent's social influence.
As an immediate consequence of Equation~(\ref{eq:geometric}),
the following lemma shows that
$I(k,\lambda,s)$ gains the same increase
by scaling up $s$ as it does by scaling up $\lambda$.

\begin{lemma}\label{lem:c_vs_i}
Let $r\in \mathbb{R}$, $r\geq 1$, $rs<1$ and $r\lambda<1$.
Then
$$I(k,r\lambda,s)=I(k,\lambda,rs).$$
\end{lemma}

\noindent
Note that both the weight the other agents place on the external agent and 
the combined social influence of the agents that can be influenced by the external one 
are percentages, thus the precondition in Lemma~\ref{lem:c_vs_i}.
The lemma leads to our four insight:
\begin{quote}
\textit{(A4) if the timing option is consensus,
then it is equally effective to scale up the coverage or intensity factor 
to maximise the external agent's social influence. }
\end{quote}

\noindent
Since $\lambda$ and $s$ are equally important in determining the value of $I(k,\lambda,s)$,
it remains to compare either one of them with $k$.
The following lemma shows that  $I(k,\lambda,s)$ increases more by scaling up $s$ than it does by scaling up $k$.

\begin{lemma}\label{lem:c_vs_d}
Let $r\in \mathbb{Z}$, $r\geq 2$, $rs<1$ and $r\lambda<1$.
Then
$$I(k,\lambda,rs)> I(rk,\lambda,s).$$
\end{lemma}

\begin{proof}
Since, according to Equation~(\ref{eq:geometric}),
$I(k,\lambda,rs)=1-(1-rs\lambda)^k$ and $I(rk,\lambda,s)=1-(1-s\lambda)^{rk}$,
it suffices to show $(1-rs\lambda) <(1-s\lambda)^r$.
We will prove by induction on $r$.
For the base case, we have $r=2$.
Then $(1-2s\lambda) <(1-s\lambda)^2$ follows from 
$(1-s\lambda)^2=1-2s\lambda+(s\lambda)^2$ and $s\lambda>0$.

For the induction step, suppose $(1-rs\lambda) <(1-s\lambda)^r$ holds for $r=l$, 
we need to show it also holds for $r=l+1$.
\begin{align*}
1-rs\lambda & = 1-(l+1)s\lambda\\
            &= (1-ls\lambda)-s\lambda
\end{align*}
and 
\begin{align*}
(1-s\lambda)^r &= (1-s\lambda)^{l+1}\\
              &= (1-s\lambda)^l(1-s\lambda)\\
              &= (1-s\lambda)^l-s\lambda(1-s\lambda)^l.
\end{align*}
We have by the induction hypothesis that $(1-ls\lambda)<(1-s\lambda)^l$.
Also since $(1-s\lambda)^l<1$, we have $s\lambda>s\lambda(1-s\lambda)^l$.
It follows from $(1-ls\lambda)<(1-s\lambda)^l$ and $s\lambda>s\lambda(1-s\lambda)^l$ that
$(1-ls\lambda)-s\lambda< (1-s\lambda)^l-s\lambda(1-s\lambda)^l$
which implies $1-rs\lambda<(1-s\lambda)^r$.
This completes the proof.

\end{proof}

\noindent
Note that it is meaningless to scale up the number of participation
rounds by a non-integer or by the integer 1, 
thus the precondition of Lemma~\ref{lem:c_vs_d}.
The lemma leads to our fifth insight:
\begin{quote}
\textit{(A5) if the timing option is ``consensus,''
then it is more effective to scale up the coverage or intensity factor
than the duration factor to maximise the external agent's social influence.}
\end{quote}

Finally, it is not uncommon that 
an organisation can influence everyone in its target group.
For example, this may happen if the group is relatively small 
with respect to the organisation's resource. 
In the remaining of this section, 
we deal with this special but realistic ``full coverage'' case.

In our model, full coverage means $m=n$
with which Equation~(\ref{eq:update_partial}) reduces to
\begin{equation}\label{eq:update_full}
    \mathbf{p}^{(t)}=(1-\lambda)\mathbf{T}\mathbf{p}^{(t-1)}+\Lambda
\end{equation}
\noindent
where $\Lambda$ is the $n$-dimensional vector $(\lambda,\ldots,\lambda)^\top$.

The distinguishing property of the full coverage case is that 
the timing factor does not play a part in determining
the external agent's social influence.  
The key to appreciating this property is the following easily verifiable equality.
$$(1-\lambda)\mathbf{T}(\mathbf{T}\mathbf{p}^{(t)})+\Lambda=\mathbf{T} ((1-\lambda)\mathbf{T}\mathbf{p}^{(t)}+\Lambda)$$
By Equation~(\ref{eq:DeGroot}) and (\ref{eq:update_full}), 
if the current opinion vector is $\mathbf{p}^{(t)}$
then the LHS of the equality is 
the opinion vector after two rounds of interaction 
where the external agent participates in the second round and
the RHS is the opinion vector after two rounds of interaction 
where the external agent participates in the first round.
Since this ``one-step'' change 
does not affect the opinion vector after two rounds of interaction
neither does it to the limiting opinion.
By realising that we can repeat this ``one-step'' change
for any number of times to have the $k$ rounds
of participation at any $k$ time points without affecting 
the limiting opinion, the property must hold.
This leads to our final insight obtained through function analysis.

\begin{quote}
\textit{(A6) If the external agent can influence all agents, 
then the timing factor is irrelevant to its social influence.}
\end{quote}

\section{Simulation Results}

Function analysis has its limits, 
in this section, we resort to simulations to 
obtain insights tied to the timing options.
Like it or not, simulation might be our last resort for
the start and uniform timing options.
Suppose the external agent participates in 
the first and $r$th rounds where $r$ is $2$ for ``start'' and 
an arbitrary number for ``uniform.''
Then the limiting opinion and thus the external agent's social influence can be derived 
as follows
\begin{align*}
\mathbf{p}^\infty&=\lim_{t\rightarrow\infty} \mathbf{T}^t \mathbf{p}^{(r)}\\
&=\lim_{t\rightarrow\infty} \mathbf{T}^t((\mathbf{T}-\lambda (\mathbf{T})_m)\mathbf{p}^{(r-1)}+\Lambda)\\
&=\lim_{t\rightarrow\infty} \mathbf{T}^t((\mathbf{T}-\lambda (\mathbf{T})_m)\mathbf{T}^{r-2}\mathbf{p}^{(1)}+\Lambda)\\
&=\lim_{t\rightarrow\infty} \mathbf{T}^t((\mathbf{T}-\lambda(\mathbf{T})_m)\mathbf{T}^{r-2} \Lambda
+\Lambda)\\
&=\lim_{t\rightarrow\infty} \mathbf{T}^t(\mathbf{T}^{r-1}\Lambda-\lambda(\mathbf{T}^{r-1})_m \Lambda+\Lambda)\\
&=2\lim_{t\rightarrow\infty} \mathbf{T}^t\Lambda-\lambda\lim_{t\rightarrow\infty} \mathbf{T}^t(\mathbf{T}^{r-1})_m\Lambda.
\end{align*}
Obviously, the matrix $(\mathbf{T}^{r-1})_m$ plays a part in determining $\mathbf{p}^\infty$,
meaning that any variation of the social network can result 
in a very different social influence of the external agent.
Thus, the function that expresses the social influence 
must consider all entries of $\mathbf{T}$ (as well as the $k$ timing factors).
This is an overwhelmingly large number of variables
for function analysis to be feasible.
For ``concensus'' however $(\mathbf{T}^{r-1})_m$ 
is cancelled out in the above derivation, 
making it irrelevant for determining 
the social influence (see the proof of Theorem~\ref{thm:limiting_partial}).


In this paper, all simulations are conducted with 
100 agents and with 3000 rounds of interactions. 
Additionally, 1000 simulations are conducted 
for each combination of values for the intensity, duration, coverage and timing factor.
Experimenting with various simulation settings 
shows that a larger number of agents and simulations 
makes no difference to the exhibited patterns,
which are already evident with as little as 10 agents and 10 simulations.
Moreover, the patterns do not rely on any specific factor values, 
though some values make them easy to visualise.

Our first set of simulations intends to disentangle the
varying effects of the timing options as the duration value grows.
While holding the intensity and coverage factor constant, 
for each timing option, we simulate our model for 
duration values ranging from 0 to 45 (with an increment value of 1).
In Figure~\ref{fig:duration}, 
we plot the external agent's average social influence (in 1000 simulations) induced by 
each timing option against the duration values.

\begin{figure}[ht]
\centering
\includegraphics[width=\linewidth]{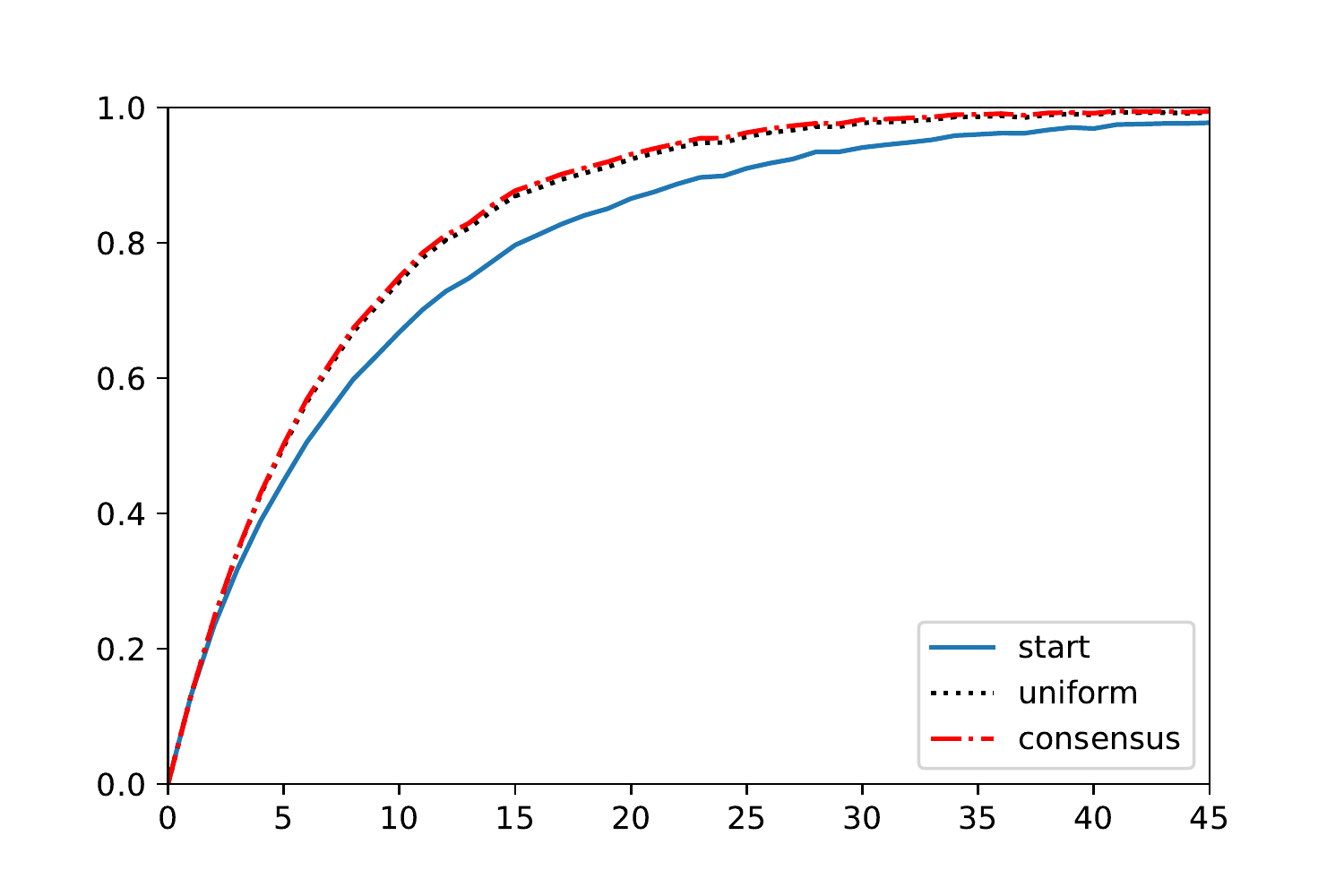}
\caption{Comparing the timing options with respect to a growing duration factor.}
\label{fig:duration}
\end{figure}

\noindent
Immediately we observe that the plot for ``consensus''
is virtually always above those of ``start'' and ``uniform,''
with a noticeable gap over the former and a tiny one over the latter.
More specifically, the gap between the plots is closing towards both ends of 
the horizontal axis and becomes negligible at the very ends. 
The plots suggest ``consensus'' and ``start''  respectively
gives the largest and smallest social influence while ``uniform'' is almost identical to 
``consensus.''
Closing of the gap, however, cannot be interpreted as suggesting 
the less relevance of the timing options with small and large duration values,
rather the pattern is enforced by the bounded nature of social influence values.
Social influence is bounded from below by 0 and from above by 1, 
thus if the duration value is sufficiently small or large, 
then the induced social influence must be close to 0 or 1 irregardless of the timing option.

\begin{figure}[ht]
\centering
\includegraphics[width=\linewidth]{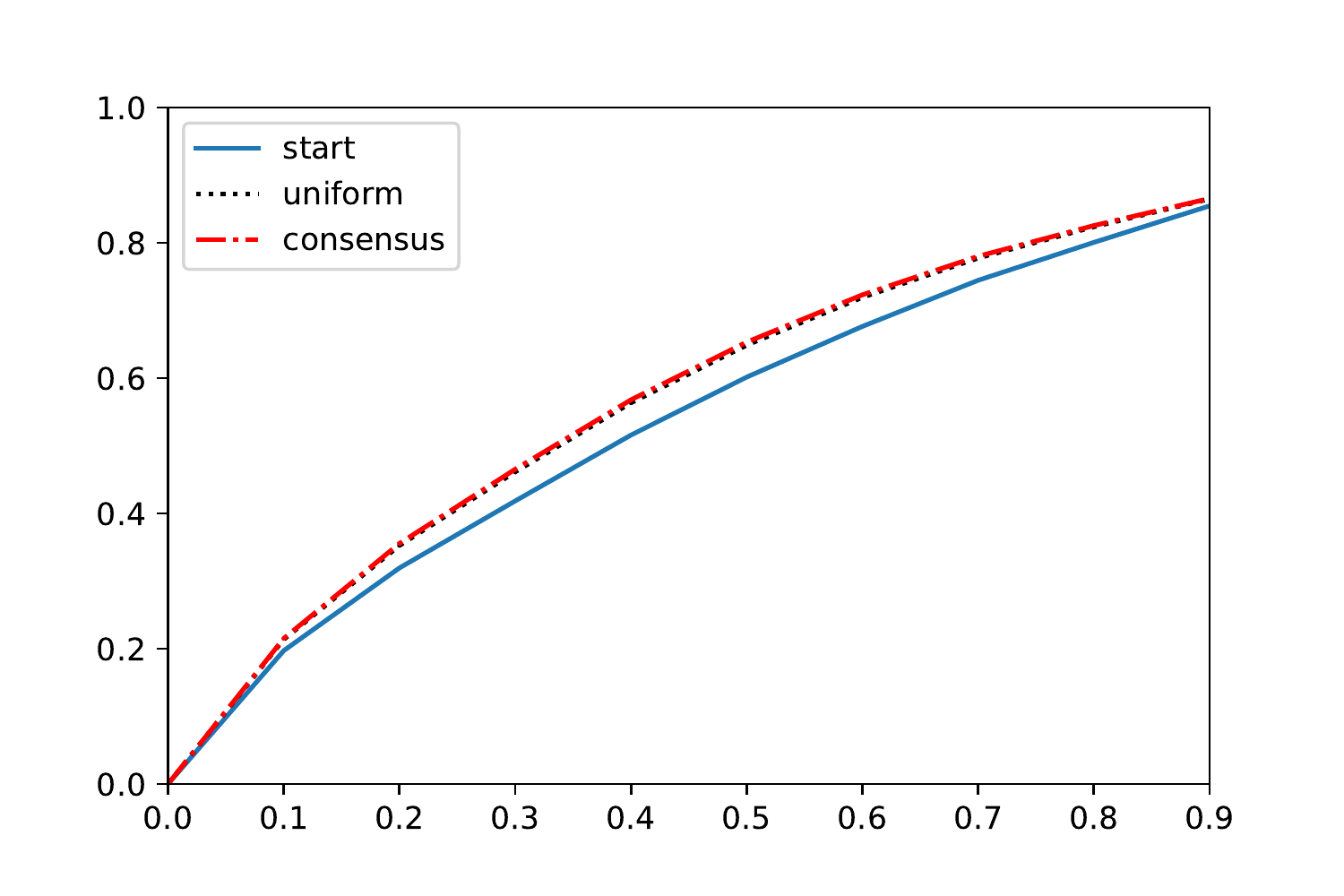}
\caption{Comparing the timing options with respect to a growing coverage factor}
\label{fig:coverage}
\end{figure}

In a similar fashion,
while holding the duration and intensity factors constant, 
for each timing option, we simulate our model for 
coverage values ranging from 0 to 0.9 (with an increment value of 0.1).
In Figure~\ref{fig:coverage}, 
the average social influence induced by each timing option is plotted 
against the coverage values.
Once again, we observe that ``consensus'' is the clear winner against ``start''
but not so much against ``uniform''.
Also the gap is closing towards both ends of the horizontal axis.
Likewise, the bounded nature of social influence plays a part in the closing of the gap.
But this time, as the coverage value grows, the gap closes so drastically 
that the plots converge at a social influence (i.e., 0.8) well below 1.
This means the growing coverage value also plays a part.
So, unlike the duration value, it is righteous to interpret the closing of the gap
as suggesting the larger the coverage value the less relevant the timing options.
In fact, we have already proved a limiting case of this pattern in (A6) which 
concludes the irrelevance of the timing factor when the coverage value is at its largest.
Hence, we obtain the following insight.

\begin{quote}
\textit{(S1) Once the coverage value passes certain threshold,\footnote{The threshold depends on the intensity and duration values.} 
then the larger the coverage value, the less relevant the timing options 
in determining the external agent's social influence.}
\end{quote}

\begin{figure}[ht]
\centering
\includegraphics[width=\linewidth]{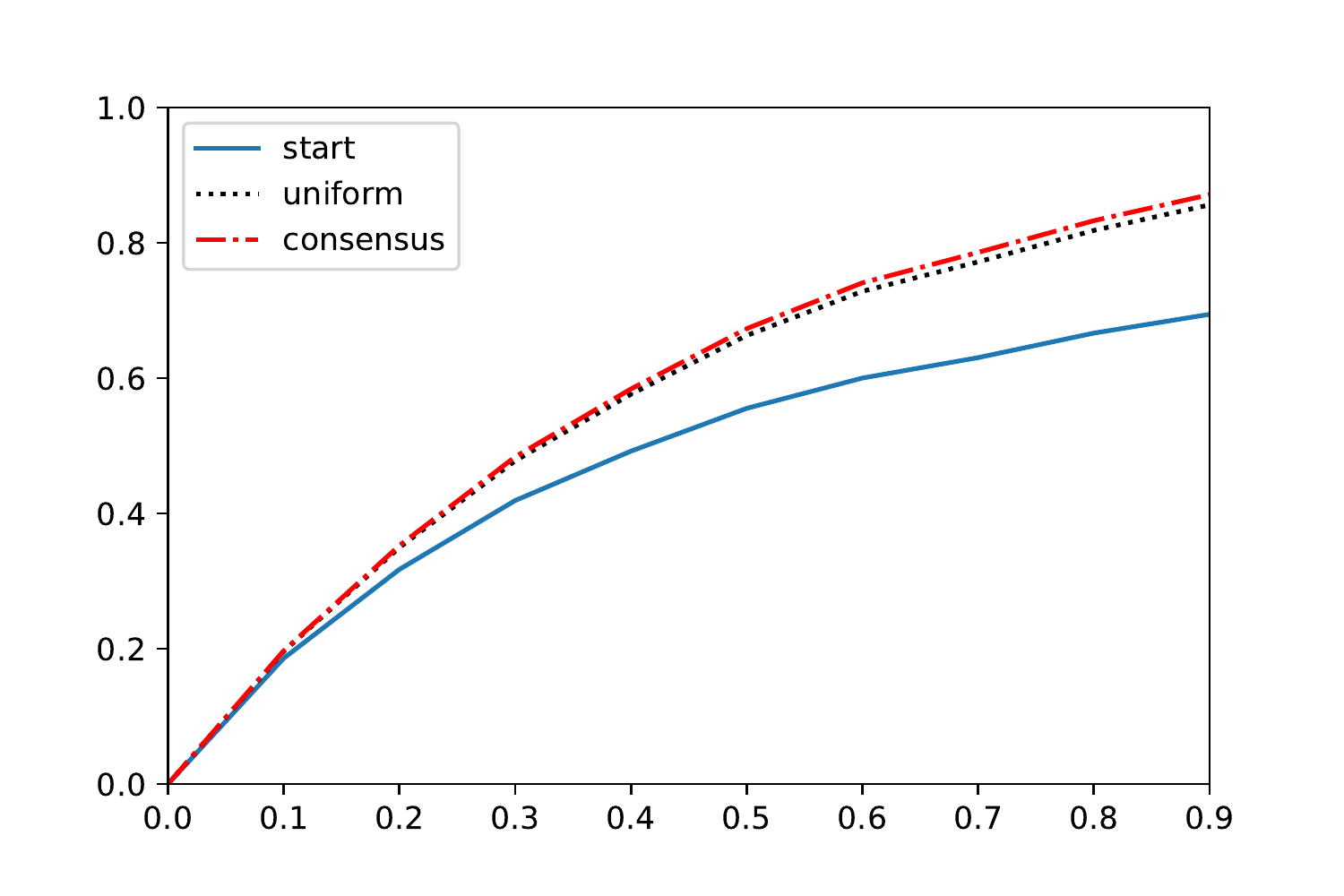}
\caption{Comparing the timing options with respect to a growing intensity factor}
\label{fig:intensity}
\end{figure}

Finally, we hold constant the duration and coverage factors and simulate our model for 
intensity values ranging from 0 to 0.9 (with an increment value of 0.1).
The corresponding plots are given in Figure~\ref{fig:intensity}.
Yet another time, we observe the superiority of ``consensus'' over
``start'' and ``uniform,'' but only slightly over the latter.
This unequivocal pattern exhibited in all three sets of simulations 
leads to the following insight.

\begin{quote}
\textit{(S2) Among the three timing options,
``consensus'' gives rise to the largest 
social influence of the external agent and ``start'' the least.}
\end{quote}

\noindent
Furthermore, the fact that ``consensus'' and ``uniform'' are almost indistinguishable 
indicates that the key is to spread out the
participation times and avoid having them in a cluster.
At last but not least, we observe that, contrary to the previous simulations,
the gap between the plots enlarges towards the end of the horizontal axis 
with larger intensity values.
Noticeably, whatever mechanism that drives this enlargement is very effective as 
it manages to do so even though the bounded nature of social influence acts   in an opposite direction.
This leads to our final insight.

\begin{quote}
\textit{(S3) The larger the intensity factor, the more 
relevant of the timing options in determining the external agent's social influence.}
\end{quote}

\section{Related Work}

\textit{Influence maximisation} is a recurring topic in studies of opinion formation, social networks
and many more.
The term is often referred to as
the algorithmic problem of selecting a predefined number of agents 
in a social network to maximise the spread of a binary opinion in an opinion diffusion process \cite{journal/TKDE/Li2018}.
\cite{conf/KDD/Domingos2001} 
are the first to pose the algorithmic problem \cite{conf/KDD/Domingos2001}.
Later on, \cite{conf/KDD/Kleinberg2003} proposed the so-called independent cascade model and gave a greedy algorithm based on submodular maximisation \cite{conf/KDD/Kleinberg2003,conf/icalp/Kempe05}.
These very influential papers have since then promoted a large amount of follow up works improving 
and extending various aspects of the selection techniques
\cite{book/Chen2013}.
Although the diffusion process and representation of opinions are different from ours,
these works are, in our terminology, endeavours within the realm of the coverage factor.
Rather than focusing on this single factor,
we move on to consider the intensity, duration and timing factors as well.
To the best of our knowledge, this is the first attempt to understand 
how the interplay between the four factors elevate and restrain the overall influencing effort.

Apart from the above works, recent years have seen a surge in the popularity 
of various forms of opinion diffusion in artificial intelligence
\cite{conf/ijcai/Brill16,conf/aamas/Grandi17,conf/tark/ChristoffG17,journal/ijar/Cholvy18,conf/ijcai/Faliszewski18,conf/aamas/Botan2019,journals/tcs/FerraioliV19,conf/aaai/Chistikov2020,journals/ai/AulettaFG20}.
Some of them also tackled influence maximisation problems \cite{conf/AAMAS/Auletta2019,conf/ijcai/Bredereck2020}.
Unlike ours and the aforementioned ones, they study how the sequence of opinion update
affects the influencing effort.
\cite{conf/AAMAS/Auletta2019} considered the case of three alternative opinions in an asynchronous mode of opinion update
and investigated the sequence of updates that maximises the spread of an opinion.
\cite{conf/ijcai/Bredereck2020} continued the work on spread-maximising sequences
and provided upper and lower bounds on the length of such sequences among other results.

\section{Conclusion}

In this paper, we generalised the DeGroot model of opinion formation 
to allow a temporary participant. 
We articulated four factors namely, duration, intensity, coverage and timing 
that dominate the temporary participant's influencing effort
and incorporated them into the opinion formation process.
Through function analysis and simulation, 
we revealed the degree of importance and interplay between the factors 
which lead to crucial insights of influence maximisation.

In summary, the temporary participant ought to adopt the consensus timing option and 
focus its resources on the intensity and coverage factor.
The insights may aid organisations and individuals to make better 
strategic choices when facing a limited resource to maximise their influence. 
Direction for future work is to investigate the case of multiple external agents. 

%
%
%

\bibliographystyle{plain}
\bibliography{opinion}

\end{document}